\documentclass[10pt, conference]{ieeeconf}      
\IEEEoverridecommandlockouts                              
\usepackage{cite}\usepackage{hyperref}
\usepackage{amsmath,amssymb,amsfonts}
\usepackage{algorithmic}
\usepackage{graphicx}
\usepackage{textcomp}
\usepackage{amsmath,amssymb,bm,bbm,mathrsfs,amscd}
\usepackage{calc}
\usepackage{color}
\usepackage{dsfont}
\usepackage{graphicx}
\usepackage{epstopdf}
\usepackage{epsfig}
\usepackage{tikz}
\usepackage{amsmath}
\usepackage{amsfonts}
\usepackage{amssymb}
\usepackage{rotating}
\usepackage{mathtools}
\usepackage{color}
\usepackage{subfig}
\usepackage{enumerate,pgfplots}
\newtheorem{theorem}{Theorem}
\newtheorem{definition}{Definition}

\newtheorem{example}{Example}
\newtheorem{remark}{Remark}

\newcommand{\ba}{\begin{array}}
\newcommand{\ea}{\end{array}}

\newcommand{\be}{\begin{equation}}
\newcommand{\ee}{\end{equation}}

\newcommand{\ds}{\displaystyle}

\newcommand{\mc}{\mathcal}

\newcommand{\Z}{\mathbb{Z}}

\newcommand{\cost}{h}
\def\1{\boldsymbol{1}}

\newcommand{\E}{\mathbb{E}}

\newcommand{\R}{\mathbb{R}}

\def\Z{\mathbb{Z}}

\def\E{\mathbb{E}}
\def\R{\mathbb{R}}

\def\1{\boldsymbol{1}}
\def\0{\boldsymbol{0}}

\tikzstyle{v_c}=[circle, draw,inner sep=2pt, minimum width=12pt, color=blue]
\tikzstyle{v_a}=[circle, draw,inner sep=2pt, minimum width=12pt, color=red]
\tikzstyle{edge} = [draw,thick,-,font=\small ]
\tikzstyle{label} = [draw,fill=black,font=\normalsize]

\def\BibTeX{{\rm B\kern-.05em{\sc i\kern-.025em b}\kern-.08em
	T\kern-.1667em\lower.7ex\hbox{E}\kern-.125emX}}
\title{\LARGE \bf Optimal interventions in opinion dynamics\\ on large-scale, time-varying, random networks}
\author{Leonardo~Cianfanelli, 
Giacomo~Como, Fabio Fagnani, Asuman Ozdaglar, Francesca Parise
\thanks{Leonardo Cianfanelli, Giacomo Como, and Fabio Fagnani are with the  Department of Mathematical Sciences ``G.L.~Lagrange,'' Politecnico di Torino, 10129 Torino, Italy  (e-mail: {\{leonardo.cianfanelli;\,giacomo.como;\,fabio.fagnani;\}@polito.it}). G.~Como is also with the Department of Automatic Control, Lund University, 22100 Lund, Sweden.  Asuman Ozdaglar is with the Department of Electrical Engineering and Computer Science at the Massachusetts Institute of Technology, Cambridge, MA, USA (e-mail:asuman@mit.edu).	Francesca Parise is with the School of Electrical and Computer Engineering, Cornell University, Ithaca, NY, USA (e-mail: fp264@cornell.edu). This project was partially supported by the Italian Ministry of University and Research under the PRIN project ‘‘Extracting essential information and dynamics from complex networks’’, grant  no. 2022MBC2EZ.}
}

\begin{document}

\maketitle
\thispagestyle{empty}
\pagestyle{empty}


\begin{abstract}	
We consider two optimization problems in which a planner aims to influence the average transient opinion in the Friedkin-Johnsen dynamics on a network by intervening on the agents' innate opinions. Solving these problems requires full network knowledge, which is often not available because of the cost involved in collecting this information or due to privacy considerations. For this reason, we focus on intervention strategies that are based on statistical instead of exact knowledge of the network. We focus on a time-varying random network model where the network is resampled at each time step and formulate two intervention problems in this setting. We show that these problems can be casted into mixed integer linear programs in the \emph{type space}, where the type of a node captures its out- and in-degree and other local features of the nodes, and provide a closed form solution for one of the two problems. The integer constraints may be easily removed using probabilistic interventions leading to linear programs. Finally, we show by a numerical analysis that there are cases in which the derived optimal interventions on time-varying networks can lead to close to optimal interventions on fixed networks.
\end{abstract}

\section{Introduction} We focus on opinion dynamics in a social network in which individual agents update their opinions based on information exchange with the neighbors. We study efficient design of interventions modeled as the ability to influence the innate opinion of the agents (for example by educational efforts or by providing external information). Unfortunately, the design of these types of interventions in large-scale networks presents several issues. First, the full network structure may be not available because of the cost involved in collecting this information or due to privacy considerations. Second, even when the full network structure is known, the design of optimal interventions scales with the network size and may be costly to compute in large-scale systems.

In this paper we are interested in considering alternative solutions based on local statistical information regarding the nodes instead of exact network information. Specifically, we formulate two optimization problems in which the planner aims to maximize the weighted transient average opinion (e.g., this may be relevant when the opinion measures the willingness of the agents to contribute to some public good) in the Friedkin-Johnsen model \cite{friedkin1999social} while reducing the intervention cost, or subject to a budget constraint. We assume that the planner knows the network \emph{statistics}, namely, the joint empirical frequency of the agent in-degrees, out-degrees, costs for the intervention, initial and innate opinions, and stubborness levels (namely, the agent confidence in their innate opinions), jointly referred to as the agent \emph{types}, without knowing the exact network structure. We also assume that the agent types take a small number of different values with respect to the number of nodes, that is, the support of the degree distribution of the network is small and the node innate opinions can take a discrete number of values. This assumption is satisfied for example when the planner is interested in changing the opinions of a few stubborn nodes in the network among some extreme values \cite{yildiz2013binary,acemoglu2011opinion}.

Since exact network information is not available, we assume that the planner is interested in designing interventions that work well in expectation for random networks that are consistent with the known network statistics. We focus in particular on a time-varying random network model in which the network is resampled at every time step. In this  case, we characterize the exact evolution of the expected average opinion as the output of a 1-dimensional linear recursion that depends only on the known quantities. We then rely on this 1-dimensional model to design interventions aimed to optimize the expected average opinion. Note that, in this framework, interventions need to be anonymous, that is, they depend solely on agent types (which is the only available information). With this constraint in mind, we assume that the planner can change the agent types at each time step, resulting in a new time-varying statistics. 
We then prove that the resulting optimization problems are mixed integer linear programs, whose dimension scales with the number of types instead of the network size, leading to a reduction of complexity in large-scale networks. The integer constraints may be easily removed using probabilistic interventions leading to linear programs. Additionally, we show that one of the two optimization problems can be solved in closed form. 

The assumption that the network is resampled at every time step simplifies the technical analysis rendering the network at time $t$ independent of the history up to time $t-1$ and is inspired by time-variability typically observed in large scale networks. Nonetheless, we  note that one could also be interested in designing interventions for random networks that are consistent with the given statistics but are sampled only once (at the beginning of the process) instead of  every time step. While the technical analysis of this case is left for future work, we provide a numerical example suggesting that the obtained 1-dimensional linear recursion could also be used to approximate the opinion evolution in the case of fixed networks, in the limit of large populations. In fact, our numerical study suggests that this approximation could hold not just in expectation but also with high probability. Similarly, our numerical study suggests that the intervention designed by using the procedure for time-varying random networks performs well even in the case of a fixed network sampled once.

Overall, the proposed approach has several advantages. First, it only requires knowledge of the network statistics instead of the exact network. Second, it works in the case of time-varying networks, which could be an advantage for large networks subject to time-variability. Third, it scales with the number of types instead of the network size, allowing for the design of interventions in very large networks. 

Our paper is related to two strands of literature. 
First, it relates to papers studying interventions in opinion dynamics over fixed deterministic networks \cite{yi2021shifting,gionis2013opinion,he2021dynamic,sun2023opinion,ancona2022model}.
We also note that our work is related to papers studying interventions in network games (see e.g. \cite{ballester2006s,galeotti2017targeting})   since, under suitable assumptions, the best-response dynamics of such games are related to  the dynamics studied in this paper. The main difference from the  previously cited papers and our work is that we assume that exact network information is not available. 
Second, our paper relates to a rich strand of literature that studies network processes and interventions in large networks by exploiting statistical network information under different random network models (including for example stochastic block models,   configuration models and graphons). See, among others,  \cite{Dasaratha2017,graphons} for centrality measures, \cite{galeotti2010network,parise2023graphon,aurell2022finite,carmona2022stochastic,caines2021graphon} for network and mean field games, \cite{lelarge2012diffusion,akbarpour2018just,erol2023contagion,sadler2020diffusion,jackson2025behavioral} for contagion/diffusion processes and \cite{gao2019graphon} for linear  dynamical systems. Among these, our paper is mostly related to \cite{golub2012homophily}, which studies the speed of learning in opinion dynamics over random networks sampled from a stochastic block model and to \cite{rossi2017threshold,messina2024optimal}, in which contagion dynamics and interventions are studied over networks generated from a configuration model. The main difference of our work from the ones above is the focus on Friedkin-Johnsen opinion dynamics and the corresponding interventions.

The paper is organized as follows. The model and the optimization problems under exact network knowledge are introduced in Section \ref{sec:2}. In Section \ref{sec:meanfield}, we propose and analyze the dynamics in which the network is randomly resampled at each time step, and formulate the optimization problem for the expected dynamics. Section \ref{sec:solutions} analyzes the properties of such optimizations problems and provide their solutions. Section \ref{sec:heuristic} discusses possible extensions to the case of dynamics without network resampling. Finally, Section \ref{sec:conclusion} summarizes the contribution and discusses future research.

\noindent \textit{Notation:}
We let $\R, \R_+, \Z, \Z_+$ denote the set of reals, non-negative reals, integers, and non-negative integers, respectively. $\0$ and $\1$ denote the vectors of all zeros and all ones. The indicator function is denoted $\mathbb I$. For a probability distribution $p$ in $[0,1]^\Omega$ over a finite set $\Omega$ and a vector $x$ in $\R^\Omega$, we let $\langle p,x \rangle = \sum_{w \in \Omega} p_w x_w$.
\section{Dynamics and interventions with exact network knowledge}\label{sec:2}
\subsection{Network and dynamics}
We model the network of agents by a finite directed multigraph $\mathcal{G} = (\mathcal{V}, \mathcal{E}, W)$, where $\mathcal{V}$ is a node set, $\mathcal{E}$ is a set of directed links, and $W$ in $\Z^{\mc V \times \mc V}_+$ is the adjacency matrix, whose $(i,j)$-th entry $W_{ij}$ denotes the number of links from $i$ to $j$, measuring the influence of node $j$ on node $i$. Let $n = |\mathcal{V}|$ denote the network size, $l = |\mc E|$ denote the cardinality of the link set and
let $\bar k = W\mathbf{1} > \0$ and $\bar d = W'\mathbf{1} > \0$ denote the out- and in-degree vectors, respectively.

Every agent is endowed with a stubborness level and an innate opinion, collected in two vectors $\bar a$, $\bar c$ in $[0,1]^{\mc V}$, respectively. Furthermore, every agent $i$ is endowed with a time-varying opinion $Z_i(t)$ in $[0,1]$, with $t$ in $\Z_+$, and the agents' opinions are collected into the opinion vector $Z(t)$ in $\mathds [0,1]^{\mc V}$.
Given an initial condition $Z(0)$ in $[0,1]^{\mc V}$, the classic discrete-time Friedkin-Johnsen (FJ) dynamics are
\begin{equation}\label{eq:fj}
	Z_i(t)=(1-\bar a_i)\frac{1}{\bar k_i} \sum_{j \in \mc V} W_{ij} Z_j(t-1)+\bar a_i \bar c_i, \quad \forall i \in \mc V\,,
\end{equation}
for every positive time $t$ in $\Z_+$.

\subsection{Intervention and optimization problems}
To influence the dynamics in \eqref{eq:fj}, we assume that a planner can design dynamical interventions that modify the agent innate opinions, with the goal of maximizing a weighted average opinion within a positive time horizon $T$ in $\Z$. For an intervention $u: 1,\cdots,T \to [0,1]^{\mc V}$, the corresponding controlled Friedkin-Johnsen dynamics are
\be\label{eq:dyn_int_fixed}
\!Z^{(u)}_i(t) = (1-\bar a_i) \frac{1}{\bar k_i} \sum_{j \in \mc V} W_{ij} Z^{(u)}_j(t-1)+\bar a_i \big(\bar c_i + u_i(t)\big),
\ee
for every time $t = 1,\cdots, T$ and node $i$ in $\mc V$. 
Since the innate opinions after the intervention must belong to $[0,1]$, an intervention is feasible only if
\be\label{eq:U}
u_i(t) \in [0,1-\bar c_i]\,, \quad \forall i \in \mc V, \ \forall t = 1,\cdots,T\,.
\ee
Now, let
\be\label{eq:avg_micro}
\bar Z^{(u)}(t) = \frac 1n \sum_{i \in \mc V} Z_i^{(u)}(t)
\ee
denote the average opinion at time $t$. Moreover, let $\theta_t \ge 0$ measure the importance of the average state at time $t$ in the objective function. For example, a uniform $\theta$ means that the planner is interested in maximizing the time-average of the average opinion up to the time horizon, whereas $\theta_T = 1$ and $\theta_t = 0$ for every $t \neq T$ means that the planner is interested in maximizing the average opinion at the final time. Let $\bar\cost_{i,t}(u_i(t))$ denote the cost function for applying intervention $u_i(t)$ on node $i$ at time $t$. We then consider two optimization problems.
In the first problem, we assume that the planner has a finite budget $B>0$, so that \be\label{eq:budget_micro}
\sum_{t=1}^T \sum_{i \in \mc V} \bar \cost_{i,t}(u_i(t)) \le B\,,
\ee
resulting in the optimization problem
\be\label{eq:prob_micro_budget}
\begin{aligned}
	\underset{u: 1,\cdots,T \to [0,1]^{\mc V}}{\max} \
	& \sum_{t = 1}^T \theta_t \bar Z^{(u)}(t) \\[6pt]
	\text{s.t.} \ & \eqref{eq:dyn_int_fixed}-\eqref{eq:budget_micro}\,.\\[8pt]
\end{aligned}
\ee
In the second problem, the cost for the intervention is included in the objective function and there is no budget constraint, namely,
\be\label{eq:prob_micro}
\begin{aligned}
	\underset{u: 1,\cdots,T \to [0,1]^{\mc V}}{\max} \
	& \sum_{t = 1}^T \theta_t \bar Z^{(u)}(t) - \sum_{t=1}^T \sum_{i \in \mc V} \bar \cost_{i,t}(u_i(t)) \\[6pt]
	\text{s.t.} \ & \eqref{eq:dyn_int_fixed}-\eqref{eq:avg_micro}\,.
\end{aligned}\\[6pt]
\ee
Note that in Problems \eqref{eq:prob_micro_budget} and \eqref{eq:prob_micro} we assumed that the central planner can modify the innate opinions of each agent in the network. One could also consider a variant in which only specific agents can be targeted by adding additional constraints.

\subsection{A restriction on agent types}
In the following, we will restrict our attention to networks in which nodes can be divided into a finite number of types. Specifically, we  associate to every node a \emph{type} $\omega_i \in \Omega$ that defines its in-degree $\bar d_i$, out-degree $\bar k_i$, stubborness level $\bar a_i$, innate opinion $\bar c_i$, initial condition $Z_i(0)$ and cost functions $\bar h_{i,t}(\cdot)$.  Let $\mc V_w=\{ i \in \mc V \mid w_i = w\}$ denote the set of nodes of type $w\in \Omega$. We then define $d,k,a,c,s,h$ in $\R^\Omega$ such that, for every $w$ in $\Omega$,
$$
d_{w} = \bar d_i, \quad k_{w} = \bar k_i, \quad a_{w} = \bar a_i, \quad c_{w} = \bar c_i\,,
$$
$$
s_{w} = Z_i(0)\,, \ h_{w,t}(\cdot) = \bar h_{i,t}(\cdot)\,, \quad \forall i \in \mc V_{w}, \ \forall t = 1,\cdots,T\,.
$$
The vector of empirical type frequencies $p$ in $[0,1]^{\Omega}$
$$
p_w = |\mc V_w|/n\,, \quad \forall w \in \Omega
$$
is referred to as  \emph{network statistics}.

\begin{remark}\label{remark:countable}
The assumption that the set of types is finite is particularly suitable for example for the French-DeGroot model with stubborn nodes. This is a particular case of the FJ model in which the nodes are classified in two subsets based on their stubborness level, that is either $0$ for regular nodes, or $1$ for stubborn nodes, so that $\bar a_i\in\{0,1\}$. In many practical cases, innate opinions may also be modeled as taking only a finite number of values, for example the assumption $\bar c_i\in \{0,0.5,1\}$ could be used to distinguish between extremists and neutral agents. In such cases it is reasonable to expect the number of types $|\Omega|$ to be much smaller than the number of agents $n$ \cite{yildiz2013binary}.
\end{remark}

\subsection{The issue for large networks}
It is important to stress that Problems \eqref{eq:prob_micro_budget}-\eqref{eq:prob_micro} are convex programs if and only if the cost functions $\bar h_{i,t}$ are convex. Even more importantly, their solution requires full network knowledge, which is often an unrealistic assumption, and their complexity scales with the network size.

To overcome these issues, we  propose an alternative intervention method based on statistical instead of exact network information, under the assumption that the network of interest is a realization of a known time-varying random graph  model consistent with the network statistics introduced in the previous subsection.

\section{Dynamics and formulation without exact network knowledge}\label{sec:meanfield}
In the rest of the paper, we shall assume that the planner knows the agent types and the network statistics generating the network, but does not have full network knowledge, 
and suggest the following  intervention  procedure. First, we assume that since exact network information is not available, the planner is interested in generating interventions that work well in expectation for networks sampled from a generating model based on the available network statistics. Note that interventions in this context need to be based only on agent types as that is the only information available. In Section~\ref{sec:dynamics} we introduce the specific random graph model used in this paper (which is time-varying in the sense that a new network is generated at each time step), the type of interventions considered, and the corresponding FJ dynamics.  In Section~\ref{sec:exp_dyn}  we  prove that the expected average opinion of the FJ dynamics on the considered time-varying random network model is given by the output of a 1-dimensional linear recursion that depends on the network statistics after the intervention. Finally, in Section \ref{sec:problem_ref}, we formulate two   optimization problems for such expected dynamics over time-varying random networks corresponding to the two optimization  Problems \eqref{eq:prob_micro_budget}-\eqref{eq:prob_micro} (which were formulated for fixed and known networks).

\subsection{The time-varying random network model with resampling}\label{sec:dynamics}

\subsubsection{Interventions}
Since without exact network knowledge the planner cannot target specific agents, we consider interventions in which the planner can  just change the agent types. We describe such interventions with a variable  $\beta:\{1,\cdots,T\} \to [0,1]^{\Omega \times\Omega}$, whose entry $\beta_{ww'}(t)$ denotes the fraction of agents whose type is modified from $w$ to $w'$ at time $t$. For a type $w$ in $\Omega$, let
\begin{align*}
\Omega(w) := \{w' \in \Omega: \ &d_{w} = d_{w'}, a_w = a_{w'}, k_w = k_{w'}, \\
&s_w=s_{w'}, h_{w,t}=h_{w',t}\}
\end{align*}
denote the set of types $w'$ that have all the same parameters as $w$ except (possibly) for the innate opinion. For every time $t=1,\cdots,T$, we assume that the intervention $\beta$ must satisfy the mass conservation constraint
\be\label{eq:constraints1}
\sum_{w' \in \Omega} \beta_{ww'}(t) = p_w\,,
\ee
for every $w$ in $\Omega$, and the constraint 
\be\label{eq:constraints2}
\beta_{ww'}(t) = 0\,, \quad \forall w' \notin \Omega(w)\,,
\ee
for every $w,w'$ in $\Omega$, which ensures that the intervention modifies only the agents' innate opinions (consistently with the model in Section \ref{sec:2}).
Note that such intervention generates a new network statistics at every time $t$, which we denote by \be\label{eq:pbeta}\ba{rcl}
\!\!\!\!\!\! p^{(\beta)}_w(t) = \!\ds \sum_{w' \in \Omega} \beta_{w'w}(t), 
\ea\ee
for every $w$ in $\Omega$.

\subsubsection{Well posedness}

For a given probability distribution $p$ on $\Omega$ and finite set $\mc V$ of cardinality $n=|\mc V|$, a network with node set $\mc V$ and statistics $p$ exists if and only if
\be\label{consistent-1moment}
\langle p,d \rangle = \langle p,k \rangle \quad \text{and} \quad p_w n \in \Z_+\,,\forall w \in\Omega.
\ee

Moreover, to be implementable in a  network with $n$ agents,  an intervention $\beta$ should satisfy the constraint
\be\label{consistent-integer}n\beta_{ww'}(t)\in\Z_+\,,\qquad\forall w,w' \in\Omega\,, \ \forall t = 1,\cdots,T\,,\ee
since the number of nodes $n_{ww'}(t)$ of type $w$ whose type is modified to $w'$ at time $t$ due to the intervention needs to be integer. An intervention $\beta$ is termed \emph{feasible} if it satisfies \eqref{eq:constraints1}, \eqref{eq:constraints2} and \eqref{consistent-integer}. A triple $(n,p,\beta)$ of a positive integer $n$, a statistics $p$, and an intervention $\beta$ is termed \emph{compatible} if conditions \eqref{eq:constraints1}, \eqref{eq:constraints2}, \eqref{consistent-1moment},\eqref{consistent-integer}
are satisfied, namely, if $\beta$ is feasible and a network with size $n$ and statistics $p$ exists.

\subsubsection{The sampling procedure}

We consider the following time-varying random network model.

\begin{definition}\label{def:resampling}
	
	 For a compatible triple $(n,p,\beta)$, a time-varying network and an intervention are constructed as follows.
	 First, a type is assigned to each agent $i\in\mc V$ consistently with $p$. Then,
	for every time $t = 1,\cdots,T$:
	\begin{enumerate}
	\item (Links) for every node $i$ in $\mc V$, draw $\bar k_i$ half-links stemming from it, and $\bar d_i$ half-links entering it. Then, match uniformly at random the $l$ half-links stemming from the nodes with the half-links entering the nodes. This defines the adjacency matrix $W(t)$, whose $(i,j)$-th entry $W_{ij}(t)$ counts the number of half-links stemming from $i$ that are matched with half-links entering $j$. $W(t)$ defines  a link set $\mc E(t)$. Note that this is equivalent to generating a network from a configuration model ensemble \cite{newman2018networks};
	\item (Interventions) For every $w$ in $\Omega$, split $\mc V_w$ in $|\Omega|$ random disjoint subsets $\mc V_{ww'}(t)$ of size $n\beta_{ww'}(t)$
	and define
	\be\label{eq:time_innate}
	u^{(\beta)}_i(t) = c_{w'}-c_w\,, \quad \bar c_i^{(\beta)}(t) = c_w + u^{(\beta)}_i(t) = c_{w'}\,, 
	\ee
	for every node $i$ in $\mc V_{ww'}(t)$ and types $w,w'$ in~$\Omega$.
	\end{enumerate}
\end{definition}
\begin{remark}\label{remark:resampling}
Note that $W(t)$ is independent of $W(\tau)$ and $u^{(\beta)}(t)$ is independent of $u^{(\beta)}(\tau)$ for every $t \neq \tau$.
\end{remark}
\begin{remark}\label{remark:design_u}
Let \be
\!\!\! n_{ww'}(t):=|\{i \in \mc V_w: u_i^{(\beta)}(t) = c_{w'} - c_w\}|
\ee
be the number of agents switched from type $w$ to $w'$ at time~$t$. Note that according to the procedure in Definition \ref{def:resampling} 
\be\label{eq:nww} 
n_{ww'}(t) = n \beta_{ww'}(t).
\ee
As an alternative, one could define the interventions in step 2) of Definition \ref{def:resampling} to be probabilistic by assuming that every node $i$ in $\mc V_w$ has probability $\beta_{ww'}(t) / p_w$ of becoming of type $w'$. This would imply that $n_{ww'}(t)$ is a random variable with $\E[n_{ww'}(t)] = n\beta_{ww'}(t)$. All subsequent results can be adapted easily to this case.
\end{remark}

\subsubsection{FJ dynamics}
 
For networks generated according to the time-varying random network model above, we define the following FJ dynamics: $Y^{(\beta)}_i(0)=Z_i(0)$ and
\be\label{eq:dyn_int}
\!Y^{(\beta)}_i(t) = (1-\bar a_i) \frac{1}{\bar k_i} \sum_{j \in \mc V} W_{ij}(t) Y^{(\beta)}_j(t-1)+\bar a_i \big(\bar c_i + u^{(\beta)}_i(t)\big),
\ee
where we used $Y^{(\beta)}(t)$ in $[0,1]^{\mc V}$ to denote the opinion vector at time $t$ under the Friedkin-Johnsen dynamics on the resampled network. Let also
\be\label{eq:X} \!\!\!\! \bar{X}^{(\beta)}(t)=\frac 1l \sum_{i \in \mc V} \bar d_i{Y_i^{(\beta)}(t)}\,, \ \bar{Y}^{(\beta)}(t)=\frac{1}{n} \sum_{i \in \mc V}{Y^{(\beta)}_i(t)}\,,\ee
denote respectively the average state weighted by the in-degree, and the average state.
Note that, since the network and interventions in this framework are stochastic, $\bar X^{(\beta)}(t)$ and $\bar Y^{(\beta)}(t)$ are random objects. The objective of the next section will be to study the evolution of their expectations,
\be\label{eq:x}
x^{(\beta)}(t) = \E_{1:t}[\bar X^{(\beta)}(t)]\,, \quad y^{(\beta)}(t) = \E_{1:t}[\bar Y^{(\beta)}(t)]\,,
\ee 
where the expected value is taken over the sequence of random networks and interventions generated until time $t$ as per Definition \ref{def:resampling}.

\subsection{Analysis of the expected dynamics}\label{sec:exp_dyn}
The next result proves that $x^{(\beta)}(t)$ is determined by a 1-dimensional linear recursion, and $y^{(\beta)}(t)$ may be seen as the output of such recursion. For this analysis, it is useful to refer to links pointing to nodes of type $w$ as \emph{links of type $w$}. Accordingly, we 
let 
\be\label{eq:q}
q_{w} = \frac{p_w d_w}{\langle p,d\rangle} = \frac{n}l p_wd_w
\ee
denote the initial fraction of links of type $w$ and
\be\label{eq:qbeta}\ba{rcl}
 q_w^{(\beta)}(t) = \!\ds \sum_{w' \in \Omega} \frac{d_{w} \beta_{w'w}(t)}{\langle p,d \rangle}
\ea\ee
be the same statistics after the intervention at time $t$.

\begin{theorem}\label{thm:dynamic}
	Let $\bar Y^{(\beta)}(t)$ be the average opinion of the FJ model on a resampled network generated by the triple $(n,p,\beta)$, and let $x^{(\beta)}(t), y^{(\beta)}(t)$ be the expected values as defined in \eqref{eq:x}. Let $\psi: [0,1] \to [0,1]^{\Omega}$  be
defined by
\be\label{eq:psi}
\psi_w(x) = (1-a_w) x + a_w c_w\,, \quad \forall w \in \Omega\,.
\ee
Then,
\begin{equation}\label{eq:recursion}
	\begin{cases}
		x^{(\beta)}(t)=\langle q^{(\beta)}(t),\psi(x^{(\beta)}(t-1))\rangle \\[4pt]
		y^{(\beta)}(t)=\langle p^{(\beta)}(t),\psi(x^{(\beta)}(t-1)) \rangle\,,
	\end{cases}
\end{equation}
for every $t=1,\cdots,T$, with $x^{(\beta)}(0) = x_0 = \langle q,s \rangle$.
\end{theorem}
\begin{proof}
Let $W_{ij}(t)$ be the number of links formed from node $i$ to node $j$  in the resampled network at time $t$ and $P_{ij}(t)=\frac{1}{\bar k_i} W_{ij}(t)$. Then, \be\label{eq:Pij}\mathbb{E}[W_{ij}(t)]=\bar{k}_i\bar d_j/l, \quad \E[P_{ij}(t)]=\bar d_j/l\,.
\ee
Given the realized opinion vector at time $t-1$, the new vector of opinions has components
$$Y_i^{(\beta)} (t) =(1-\bar a_i) \sum_j P_{ij}(t) Y^{(\beta)}_j(t-1)+ \bar a_i  \bar c^{(\beta)}_i(t)\,.$$
Let $\mathbb{E}_{t|t-1}[\bar X^{(\beta)}(t)]$ be the expectation of $\bar X^{(\beta)}(t)$ at time~$t$ conditioned on all randomness up to time $t-1$.
Then, 
$$\ba{rl}
	&\mathbb{E}_{t|t-1}[\bar X^{(\beta)}(t)]= \mathbb{E}_{t|t-1}[\frac{1}{l}\sum_i \bar d_i Y^{(\beta)}_i(t)]\\[8pt]
	= & \underbrace{\mathbb{E}_{t|t-1}[\frac{1}{l}\sum_i \bar d_i (1-\bar a_i) \sum_j P_{ij}(t) Y_j^{(\beta)}(t-1) ]}_{term 1} +\\[8pt]
	+ & \underbrace{\mathbb{E}_{t|t-1}[ \frac{1}{l}\sum_i \bar d_i \bar a_i  \bar c_i^{(\beta)}(t)]}_{term 2}.
\ea$$
We now analyze the two terms separately. First, note that 
\begin{align*}
term 1 &= \frac{1}{l}\sum_w  d_w (1- a_w) \!\! \sum_{\substack{i\in\mathcal{V}_w, \\[1pt] j \in \mc V}} \mathbb{E}_{t|t-1}[ P_{ij}(t)] Y^{(\beta)}_j(t-1) \\
	&= \frac{1}{l}\sum_w  d_w (1- a_w) \sum_{i\in\mathcal{V}_w} \bar X^{(\beta)}(t-1) \\
	&=\Big[ \sum_w   \frac{n}{l} p_w d_w      {(1- a_w)}  \Big] \bar X^{(\beta)}(t-1) \\
	&=\Big[ \sum_w  q_w      {(1- a_w)}  \Big] \bar X^{(\beta)}(t-1) \\
	&=\Big[ \sum_w  q^{(\beta)}_{w}(t)     {(1- a_w)}  \Big] \bar X^{(\beta)}(t-1)\,,
\end{align*}
where we used \eqref{eq:Pij} and \eqref{eq:X} in the second equality, \eqref{eq:q} in the fourth equality, and in the last step \eqref{eq:constraints2}, \eqref{eq:q} and \eqref{eq:qbeta}, which together imply $\langle q^{(\beta)}(t), \1-a \rangle = \langle q, \1- a \rangle$ for every feasible intervention $\beta$ and time $t$. For the second term, note that
\begin{align*}
	term 2&=\mathbb{E}_{t|t-1}[ \frac{1}{l}\sum_w d_w  a_{w}  \sum_{i\in\mathcal{V}_w}   \bar c^{(\beta)}_i(t)]\\
	&=\mathbb{E}_{t|t-1}[ \frac{1}{l}\sum_w d_w  a_w  \sum_{i\in\mathcal{V}_w}\sum_{w'} c_{w'}  \mathbb{I}[ \bar c_i^{(\beta)}(t)=c_{w'} ]]\\
	&= \frac{1}{l}\sum_w  d_w a_w \sum_{w'} c_{w'}   \mathbb{E}_{t|t-1}[ \sum_{i\in\mathcal{V}_w}  \mathbb{I}[ \bar c^{(\beta)}_i(t)=c_{w'} ]]\\
	&= \frac{1}{l}\sum_w  d_w a_w \sum_{w'} c_{w'}   \mathbb{E}_{t|t-1}[n_{ww'}(t)]\\
	&=\sum_{w'} a_{w'} c_{w'} \sum_w  \frac{n}{l} d_{w'} \beta_{ww'}(t)\\
	&=\sum_{w'}a_{w'} c_{w'}  q^{(\beta)}_{w'}(t)=\sum_{w}a_{w} c_{w}  q^{(\beta)}_{w}(t)
\end{align*}
where we used \eqref{eq:nww} and \eqref{eq:constraints2} in the fifth equivalence. In particular, \eqref{eq:constraints2} implies that, if $\beta_{ww'}(t)>0$, then, $d_w = d_{w'}$ and $a_w = a_{w'}$. This proves that $\E_{t|t-1}[\bar X^{(\beta)}(t)] = \langle q^{(\beta)}(t),\psi\big(\bar X^{(\beta)}(t-1)\big) \rangle$. Using this,
$$
\begin{aligned}
x^{(\beta)}(t) & = \E_{1:t}[\bar X^{(\beta)}(t)]  = \E_{1:t-1}\big[\E_{t|t-1}[\bar X^{(\beta)}(t)]\big] \\[3pt]
& = \E_{1:t-1} \big[ \langle q^{(\beta)}(t),\psi(\bar X^{(\beta)}(t-1))\rangle \big] \\[3pt]
& = \langle q^{(\beta)}(t),\psi(\E_{1:t-1}[\bar X^{(\beta)}(t-1)])\rangle\\
& = \langle q^{(\beta)}(t),\psi(x^{(\beta)}(t-1))\rangle\,,
\end{aligned}
$$
where we used the linearity of $\psi$.
Moreover, $x_0 = \E[\bar X^{(u)}(0)] =  \sum_{i} \bar d_i Y_i^{(\beta)}(0) / l = n \sum_w p_w d_w s_w / l = \langle q,s \rangle$. Same techniques prove the second equation in \eqref{eq:recursion}, with the only difference that $\bar Y^{(\beta)}(t)$ is the unweighted average state (cf. \eqref{eq:X}), resulting in the replacement of $q^{(\beta)}$ with $p^{(\beta)}$.
\end{proof}

Theorem \ref{thm:dynamic} states that the \emph{expected} average opinion when the network is resampled at each time step according to the procedure in Definition \ref{def:resampling} is described by a 1-dimensional linear recursion with output $y^{(\beta)}(t)$. One could also be interested in studying similar dynamics when the network is sampled only once at the beginning of the process instead of every iteration (we refer to this as the \textit{fixed random network case}). The  analysis of the expected dynamics in this model is more complex because $\mathbb{E}_{t|t-1}[ W_{ij}(t)]$ would in this case be time-varying and dependent on the  history observed up to $t-1$. Despite these technicalities, the next example  suggests that as $n$ grows large and for finite $T$,  $y^{(\beta)}(t)$ could still provide a good approximation of the realized dynamics, in fact not just in expectation but with high probability.

\begin{example}\label{example}
\begin{figure}
\centering
\includegraphics[width=6cm]{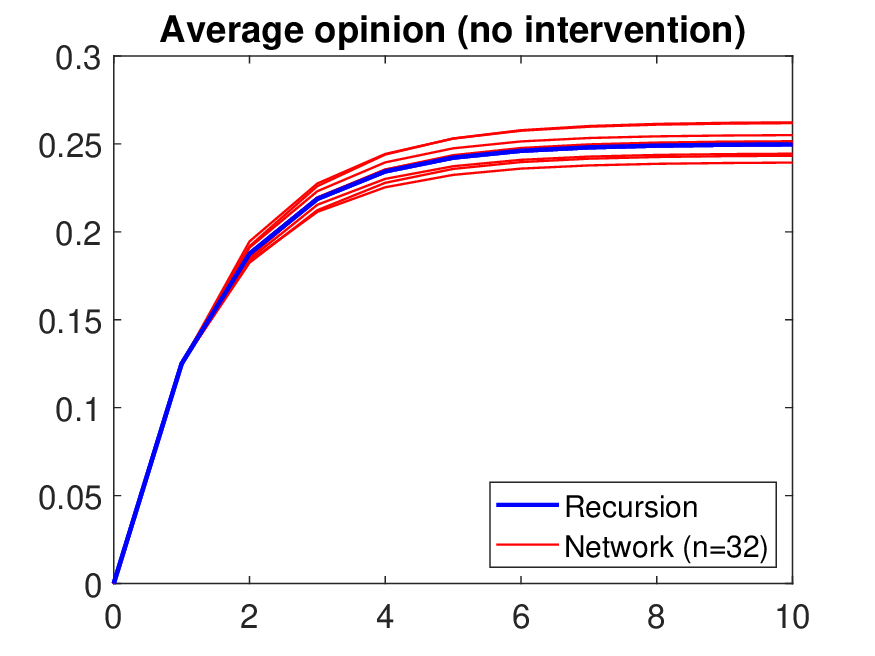}
\includegraphics[width=6cm]{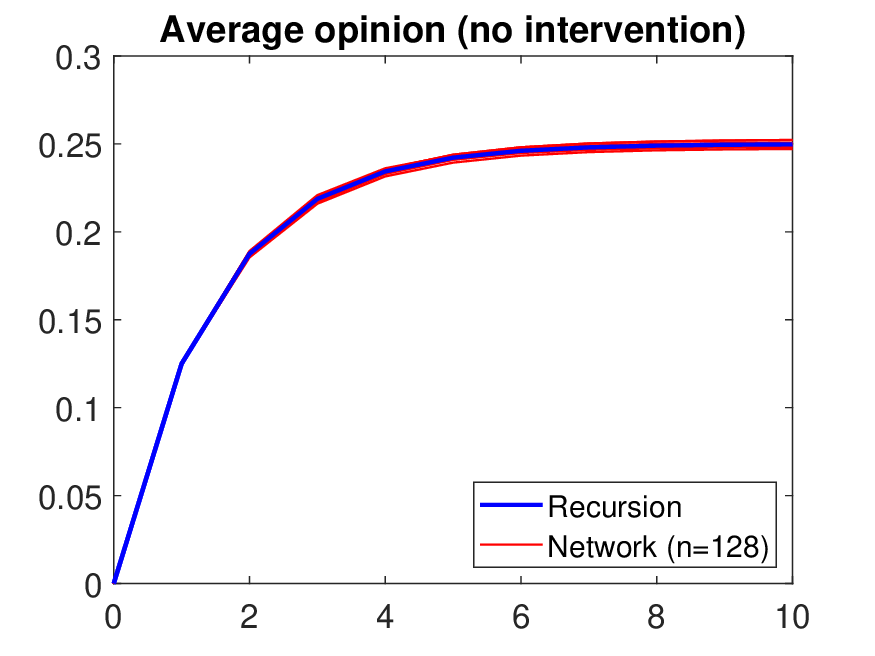}
\caption{\emph{Blue}: The average opinion predicted by recursion \eqref{eq:recursion}. \emph{Red}: The average opinion on 10 network realizations with statistics $p$ defined in Example \ref{example}, and $n=32$ (\emph{above}) and $n=128$ (\emph{below}). \label{fig:concentration}}
\end{figure}
Consider a network with $|\Omega| = 8$ types, with in-degree, out-degree, stubborness, innate opinion and initial condition
	$$
	\ba{rcl}
	k & = & (5,5,5,5,5,5,5,5)\,, \\[2pt]
	d & = & (3,7,3,7,3,7,3,7)\,, \\[2pt]
	a & = & (1,1,3,3,1,1,3,3)/4\,, \\[2pt]
	c & = & (0,0,0,0,1,1,1,1)/2\,, \\[2pt]
	s & = & (0,0,0,0,0,0,0,0)\,.
	\ea
	$$
Figure \ref{fig:concentration} illustrates (in red) the average state of the FJ dynamics \eqref{eq:fj} on $10$ fixed random networks sampled from a configuration model \cite{newman2018networks} with statistics $p = \1/8$ and $n=32$ nodes (above) vs $n=128$ nodes (below). The blue line represents the average opinion predicted by recursion \eqref{eq:recursion} for the null intervention $\beta$ such that $\beta_{ww}(t) = p_w$ and $\beta_{ww'}(t) = 0$ for every time $t$ and every pair of different types $w,w'$. Observe that, despite the number of nodes being relatively small, the average opinions in the fixed networks concentrate around the expected average opinion predicted by the recursion for the resampled network (i.e., the trajectories concentrate around the blue line as $n$ increases).
\end{example}

\subsection{Optimizing the expected dynamics}\label{sec:problem_ref}
Based on Theorem \ref{thm:dynamic}, one can define an   intervention in which $\beta$ is designed to optimize the expected dynamics.
To this end, for every $w,w'$ in $\Omega$ and $t = 1,\cdots,T$, we define $g_{ww'}(t) = n  \cost_{w,t}(c_{w'}-c_w)$, which is the cost for modifying the type of $n$ agents from type $w$ to type $w'$, so that the budget constraint \eqref{eq:budget_micro} can be written in the type space as
\be\label{eq:budget}
\sum_{t=1}^T \sum_{w,w' \in \Omega} g_{ww'}(t) \beta_{ww'}(t) \le B\,.
\ee
Problem \eqref{eq:prob_micro_budget} then becomes
\begin{equation}\label{eq:opt}\ba{cl}
	\!\!\!\!\!\underset{\beta: \{1,\cdots,T\} \to [0,1]^{\Omega \times \Omega}}{\text{max}}
	 & \ds \sum_{t=1}^T \theta_t y^{(\beta)}(t) \\[16pt]
	\text{s.t.} & \eqref{eq:constraints1}-\eqref{eq:pbeta},\eqref{consistent-integer},\eqref{eq:qbeta},\eqref{eq:recursion},\eqref{eq:budget}\,,
	\ea
\end{equation}
and problem \eqref{eq:prob_micro} becomes
\begin{equation}\label{eq:opt2}\ba{cl}
	\!\!\! \underset{\beta: \{1,\cdots,T\} \to [0,1]^{\Omega \times \Omega}}{\text{max}}
	\!\!\!\! & \ds \sum_{t=1}^T \theta_t y^{(\beta)}(t) - \sum_{t = 1}^T \sum_{w,w'} g_{ww'}(t) \beta_{ww'}(t)\\[16pt]
	\text{s.t.} & \eqref{eq:constraints1}-\eqref{eq:pbeta},\eqref{consistent-integer},\eqref{eq:qbeta},\eqref{eq:recursion}\,.
	\ea
\end{equation}
More general problems where the planner can intervene only on a subset of agents may be defined by adding constraints on $\beta$.

\section{Optimization problem solutions}\label{sec:solutions}
This section analyzes the solutions of the two optimization problems \eqref{eq:opt}-\eqref{eq:opt2}. 
To this end we define
\be\label{eq:m}
m_p = 1-\langle a,p \rangle, \quad m_q = 1-\langle a,q \rangle\,,
\ee
and, for every time $t$ and pair of types $w,w'$, we define
$$
r_{ww'}(t) = a_{w}(c_{w'}-c_{w}) \Big(\theta_t + m_p \frac{d_{w}}{\langle p,d \rangle}\sum_{\tau = t+1}^T \theta_\tau m_q^{\tau-t-1}\Big)\,.
$$
The next result states that the two optimization problems are linear programs (with integer constraints) in the space of types, and provides the explicit solution of \eqref{eq:opt2}.
\begin{theorem}\label{thm:problem}
	\emph(i) The optimization problem \eqref{eq:opt} is equivalent (up to constant terms) to the mixed integer linear program
	\begin{equation}\label{eq:new_opt}\ba{cl}
		 \underset{\beta: \{1,\cdots,T\} \to [0,1]^{\Omega \times \Omega}}{\max}
		& \ds \sum_{t = 1}^T \sum_{w,w' \in \Omega} \beta_{ww'}(t)r_{ww'}(t)\ \\[15pt]
		\text{s.t.} & \eqref{eq:constraints1},\eqref{eq:constraints2},\eqref{consistent-integer},\eqref{eq:budget}\,.
		\ea
	\end{equation}
	
\emph(ii) The optimization problem \eqref{eq:opt2} is equivalent (up to constant terms) to the mixed integer linear program 
	\begin{equation}\label{eq:new_opt2}\ba{cl}
	 \!\!\! \underset{\beta: \{1,\cdots,T\} \to [0,1]^{\Omega \times \Omega}}{\max}
	& \!\! \ds \sum_{t = 1}^T \sum_{w,w' \in \Omega} \beta_{ww'}(t)\big(r_{ww'}(t)-g_{ww'}(t)\big) \\[15pt]
	\text{s.t.} & \!\! \eqref{eq:constraints1},\eqref{eq:constraints2},\eqref{consistent-integer}\,.
	\ea
\end{equation}
Moreover, for every type $w$ and time $t$, let  $\Omega^*(w,t) \subseteq \Omega(w)$ be the non-empty set of types defined by
$$
\Omega^*(w,t) =  \underset{w' \in \Omega(w)}{\arg\max} \big(r_{ww'}(t)-g_{ww'}(t)\big)\,.
$$
Then, $\beta^*$ is a solution of \eqref{eq:new_opt2} if and only if it satisfies \eqref{eq:constraints1}, \eqref{consistent-integer} and
\be\label{eq:bstar}
\beta^*_{ww'}(t) = 0\,, \ \forall w \in \Omega\,, \ \forall t = 1,\cdots,T\,, \ \forall w' \notin \Omega^*(w,t)\,,
\ee
and the set of solutions is non-empty.
\end{theorem}\smallskip
\begin{proof}
	\emph(i) Define for simplicity $\mu$ in $\R^{\Omega}$ with entries $\mu_w = a_w c_w$	
	and note that \eqref{eq:constraints2} yields 
	$\langle p^{(\beta)}(t),a \rangle = \langle p,a \rangle$, $\langle q^{(\beta)}(t),a \rangle = \langle q,a \rangle$, and $\langle p^{(\beta)}(t), d \rangle =  \langle p,d \rangle$,
	for every feasible intervention $\beta$ and time $t$.
	Using this, \eqref{eq:recursion}, and \eqref{eq:m}, we get 
	\be\label{eq:x_beta}
	x^{(\beta)}(t) = m_q x^{(\beta)}(t-1) + \langle q^\beta(t),\mu \rangle\,.
	\ee 
	Iterating this equation,
	\be\label{eq:x_beta}
	\begin{aligned}
	x^{(\beta)}(t) &= m^t_q x_0 + \sum_{\tau=0}^{t-1} m_q^{t-1-\tau} \langle q^{(\beta)}(\tau+1),\mu \rangle\,\\
	&=m^t_q x_0 + \sum_{\tau=1}^{t} m_q^{t-\tau} \langle q^{(\beta)}(\tau),\mu \rangle\,.
	\end{aligned}
	\ee
	Note that \eqref{eq:recursion} implies $y^{(\beta)}(t) = m_p x^{(\beta)}(t-1) + \langle p^{(\beta)}(t),\mu\rangle$.
	Hence,
	\be\label{eq:calcoli}
	\!\!\ba{rcl}
	\!\! & & \ds \sum_{t=1}^T \theta_t y^{(\beta)}(t) \\[3pt] & = & \ds \sum_{t=1}^T \theta_t m_p m_q^{t-1} x_0 + \sum_{t=1}^T \theta_t \langle p^{(\beta)}(t),\mu \rangle \ + \\[3pt]
	& + & \ds \sum_{t = 1}^T \theta_t m_p \sum_{\tau = 1}^{t-1} m_q^{t-1-\tau} \langle q^{(\beta)}(\tau),\mu \rangle \\[3pt]
	& = & \ds \sum_{t=1}^T \theta_t \Big(m_p m_q^{t-1} x_0 + \sum_{w,w'} \beta_{w'w}(t) \mu_w\Big) + \\[3pt]
	& + & \ds m_p \sum_{t = 1}^T\sum_{w,w'}\frac{\beta_{w'w}(t) d_w \mu_w}{\langle p,d \rangle} \sum_{\tau = t+1}^T \theta_\tau m_q^{\tau-t-1}\,,
	\ea
	\ee
	where the last equivalence follows from the fact that for an arbitrary function $f(t,\tau)$,
	$$
	\sum_{t=1}^T \sum_{\tau=1}^{t-1} f(t,\tau) = \sum_{\tau=1}^T \sum_{t= \tau+1}^{T} f(t,\tau) = \sum_{t=1}^T \sum_{\tau = t+1}^{T} f(\tau,t)\,.$$
	From \eqref{eq:calcoli}, using $\mu_w = a_w c_w$, the fact that $m_p m_q^{t-1} x_0$ does not depend on the intervention $\beta$, and swapping $w$ and $w'$, we then get that \eqref{eq:opt} is equivalent (up to constant terms) to
	$$
	\ba{rl}
		\!\!\!\!\!\! \underset{\beta: \{1,\cdots,T\} \to [0,1]^{\Omega \times \Omega}}{\text{max}}
		& \ds \sum_{t = 1}^T \sum_{w,w'} \beta_{ww'}(t)\iota_{w'}(t)\ \\[15pt]
		\text{s.t.} & \eqref{eq:constraints1},\eqref{eq:constraints2},\eqref{consistent-integer},\eqref{eq:budget}\,,
		\ea
	$$
	with $$\iota_{w'}(t) = a_{w'} c_{w'} \Big(\theta_t + m_p \frac{d_{w'}}{\langle p,d \rangle}\sum_{\tau = t+1}^T \theta_\tau m_q^{\tau-t-1}\Big)\,.$$
	To conclude the proof, note from \eqref{eq:constraints2} that $\iota_{w'}(t)= r_{ww'}(t) +\iota_{w}(t)$ for every pair $w,w'$ such that $\beta_{ww'}(t)>0$, and
	$$
	\sum_{w,w'}\iota_{w}(t) \beta_{ww'}(t) = \sum_{w}\iota_{w}(t) p_{w}
	$$
	does not depend on the intervention $\beta$.
	
	\emph(ii) The proof of the equivalence between \eqref{eq:opt2} and \eqref{eq:new_opt2} follows the same steps as in (i). Now, observe that for every feasible $\beta$, the objective function of \eqref{eq:new_opt2} is upper bounded by
	$$
	\ba{rl}
	& \ds \sum_{t = 1}^T \sum_{w,w' \in \Omega} \beta_{ww'}(t)\big(r_{ww'}(t)-g_{ww'}(t)\big) \\[10pt]
	\le & \ds \sum_{t=1}^T \sum_{w \in \Omega} \underset{w'' \in \Omega(w)}{\max} \big(r_{ww''}(t)-g_{ww''}(t)\big) \sum_{w' \in \Omega} \beta_{ww'}(t) \\
	= & \ds \sum_{t=1}^T \sum_{w \in \Omega} p_w \cdot \underset{w'' \in \Omega(w)}{\max} \big(r_{ww''}(t)-g_{ww''}(t)\big)\,,
	\ea$$
	where the inequality follows from \eqref{eq:constraints2} and the equality from \eqref{eq:constraints1}. Note in particular that the bound is tight if and only if, for every $w$ and $t$, $\beta_{ww'}(t) = 0$ for every $w' \notin \Omega^*(w,t)$. This proves \eqref{eq:bstar}. The existence of at least a solution $\beta^*$ follows from the fact that $\Omega^*(w,t)$ is non-empty for every $w$ and $t$. Given this observation, we can always construct a solution $\beta^*$ as follows. For every $w$ and $t$, we take an arbitrary type $\tilde w_{w,t}$ in $\Omega^*(w,t)$ and  let $\beta^*_{w\tilde w_{w,t}}(t) = p_w$.
\end{proof}

\begin{remark} 
Note that, under constraint \eqref{consistent-integer}, problems \eqref{eq:new_opt} and \eqref{eq:new_opt2} are mixed integer linear programs, as $\beta$ is constrained to take a finite set of values. An alternative formulation builds on the definition of the resampled network given in Remark~\ref{remark:design_u}. In this case, constraint \eqref{consistent-integer} may be omitted, and the two problems become linear programs.
\end{remark}

Theorem \ref{thm:problem} states that the intervention problem with budget constraint may be easily solved numerically, and provides an explicit solution of the intervention problem with no budget constraint. In particular, the solutions have the following structure. At each time $t$, the innate opinion $c_w$ of all nodes of type $w$ must be increased to $c_{w'}$, where $w'$ is a type that maximizes $r_{ww'}(t)-g_{ww'}(t)$ (i.e., the benefit of the intervention minus the cost of the intervention) among those that differ from $w$ only in the innate opinion. Observe that, if the cost of the intervention on a type $w$ at time $t$ exceeds the benefit of the intervention for every alternative type $w'\in\Omega(w)$, then it is optimal not to intervene, as $r_{ww}(t) = g_{ww}(t) = 0$ for all types and times. Observe also that $r_{ww'}(t)$ is composed of two terms. 
The term $$\theta_t a_{w} (c_{w'}-c_w)$$ is the benefit that increasing an agents' innate opinion at time $t$ brings to the average state at time $t$. The term $$\frac{d_{w}}{\langle p,d\rangle}a_{w}(c_{w'}-c_w) m_p \sum_{\tau = t+1}^T \theta_\tau m_q^{\tau-t-1}$$ describes how increasing the agents' innate opinion at time $t$ influences the average state in the following times $\tau = t+1,\cdots,T$ due to the exchange of opinions in the network. Notice that this term becomes more relevant when $t$ is small with respect to the time horizon $T$, when the average stubborness weighted by the indegree $\langle q,a \rangle$ is small (therefore, $m_q$ is large), and for types with large stubborness, and large in-degree compared to average in-degree. In fact, all these factors contribute to spreading the input to the rest of the network. Finally, notice that the out-degree of the nodes does not play any role in the optimization problem.

\section{An heuristic procedure for fixed networks}\label{sec:heuristic}

The analysis and optimality guarantees in Sections \ref{sec:meanfield} and \ref{sec:solutions} refer to random networks generated according to the procedure in Definition \ref{def:resampling} with network resampling. Yet, from this analysis one could also define an heuristic procedure for influencing opinion dynamics in fixed networks by computing the solution of problems \eqref{eq:new_opt} and \eqref{eq:new_opt2} for the empirical network statistics associated with the fixed network. Such approach is suboptimal with respect to solving problems \eqref{eq:prob_micro_budget}-\eqref{eq:prob_micro} for the following  reasons. First, problems \eqref{eq:new_opt} and \eqref{eq:new_opt2} are formulated for random networks with resampling, not for fixed networks (yet this heuristic solution could be appropriate for fixed networks that can be well approximated by the random network model). Second, the intervention in problems \eqref{eq:new_opt} and \eqref{eq:new_opt2} is anonymous in the sense that the planner cannot distinguish among agents of the same type. This restricts the space of possible interventions, thus inducing some level of sub-optimality. Third, since the number of types is assumed finite, in this setting the innate opinions of the agents can be modified only to a finite set of values, in contrast with problems  \eqref{eq:prob_micro_budget}-\eqref{eq:prob_micro}, where the intervention $u_i(t)$ was assumed continuous. Despite these caveats,  the next example suggests that there are cases in which this heuristic works well.

\addtocounter{example}{-1}
\begin{example}[continued]
\begin{figure}
	\centering
	\includegraphics[width=7cm]{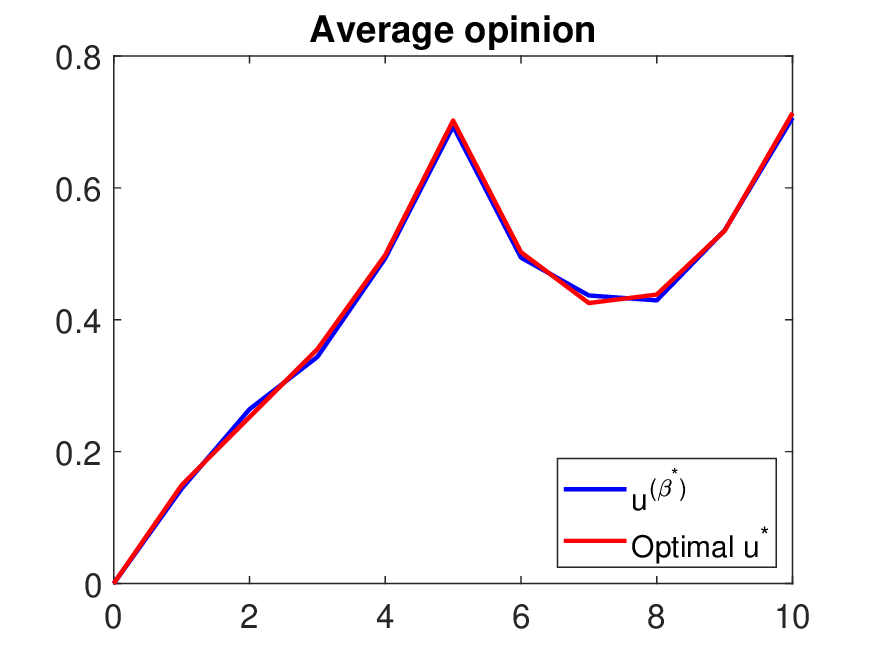}
	\caption{The average opinion $\bar Z^{(u^{(\beta)})}(t)$ under intervention $u^{(\beta)}$ (\emph{blue}) and the average opinion $Z^{(u^*)}(t)$ under the optimal intervention $u^*$ (\emph{red}), for Example \ref{example} with $n=64$ nodes. \label{fig:micromacro}}
\end{figure}
We assume that for every node $i$ and time $t=1,\cdots,10$, the intervention cost is $\bar h_{i,t}(u_i(t)) = (u_i(t))^2/n$. We consider a unitary budget $B=1$, and let $\theta_5 = \theta_{10} = 1/2$ and $\theta_t = 0$ for every $t \neq 5,10$. We then generate a fixed random network with statistics $p = \1 / 8$ and $n=64$, and find numerically the optimal solution $u^*$ of \eqref{eq:prob_micro_budget}. 
Then, we find a numerical solution $\beta^*$ of \eqref{eq:opt} using Theorem \ref{thm:problem}(i). To do this, we work under the assumption that the innate opinion of the agents may be increased with a step-size equal to $1/10$, and expand the set of types accordingly by adding new types with null empirical frequency before the intervention that differ from the original ones only in the innate opinion. We then generate a random intervention $u^{(\beta^*)}$ as by Definition \ref{def:resampling}.
Figure \ref{fig:micromacro} compares the average opinion of \eqref{eq:dyn_int_fixed} in time under intervention $u^{(\beta^*)}$, with the average opinion of \eqref{eq:dyn_int_fixed} under the exact optimal intervention $u^*$. Although $u^{(\beta^*)}$ allows only discrete variations of the agent innate opinions, the two interventions achieve a similar performance, as
$$
\sum_{t = 1}^T \theta_t \bar Z^{(u^*)}(t) = 0.6998\,, \quad \sum_{t = 1}^T \theta_t \bar Z^{(u^{(\beta^*)})}(t) = 0.7078\,,
$$
showing the validity of the heuristic approach for such fixed random network.
\end{example}

\section{Conclusion}\label{sec:conclusion} 
We formulated two intervention problems in which the planner can dynamically modify the innate opinions of the agents of a network to maximize the expected transient average opinion of  Friedkin-Johnsen dynamics over random time-varying networks. We showed that such problems can be reformulated as linear programs whose dimension depends on the number of agent types instead of the network size. Moreover, the proposed approach relies on network statistics instead of exact network knowledge. Future research lines include investigating the relation between the obtained recursion and the dynamics on fixed random networks and the application of similar techniques for optimization problems that involve higher-order moments of the network state, or to consider interventions on other features of the nodes, such as their stubborness level. Moreover, we aim to apply these techniques to other network dynamics.

\bibliographystyle{IEEEtran}
\bibliography{bib}

\end{document}